%% file: dlite_2019.tex
\documentclass{article}
\usepackage{authblk}
\usepackage{hyperref}
\usepackage[T1]{fontenc}
\usepackage{amsmath}
\usepackage{amsthm}

\newtheorem{theorem}{Theorem}
\theoremstyle{definition}

\usepackage{mathtools}

\DeclarePairedDelimiter\abs{\lvert}{\rvert}%
\DeclarePairedDelimiter\norm{\lVert}{\rVert}%

\makeatletter
\let\oldabs\abs
\def\abs{\@ifstar{\oldabs}{\oldabs*}}
\let\oldnorm\norm
\def\norm{\@ifstar{\oldnorm}{\oldnorm*}}
\makeatother

\title{DLITE: The Discounted Least Information Theory of Entropy}
\author{Weimao Ke}
\affil{wk@drexel.edu\\Drexel University}

\begin{document}
\maketitle

We propose an entropy-based information measure, namely the Discounted Least Information Theory of Entropy (DLITE), which not only exhibits important characteristics expected as an information measure but also satisfies conditions of a metric. Classic information measures such as Shannon Entropy, KL Divergence, and Jessen-Shannon Divergence have manifested some of these properties while missing others. This work fills an important gap in the advancement of information theory and its application, where related properties are desirable.

\section{Formulation}

\subsection{Least Information Theory (LIT)}

In our prior work, we proposed the Least Information Theory (LIT) to quantify the amount of entropic difference between two probability distributions \cite{Ke:2015}. Given probability distributions $P$ and $Q$ of the same variable $X$, LIT is computed by:

\begin{eqnarray}
  LIT(P, Q)
  & = & \sum_{x \in X} \int_{p_x}^{q_x} - \log p\ dp \\
  & = & \sum_{x \in X} \Big\lvert p_x (1-\ln p_x) - q_x (1-\ln q_x) \Big\rvert
\end{eqnarray}

% where $m$ is the total number of mutually exclusive inferences, i.e. the number of discrete choices for the probability distribution.
where $x$ is one of the mutually exclusive inferences of $X$, and $p_x$ and $q_x$ are probabilities of $x$ on the $P$ and $Q$ distributions respectively.

For any probabilities $p$ and $q$, let:

\begin{eqnarray}
lit(p, q) & = & \Big\lvert p (1-\ln p) - q (1-\ln q) \Big\rvert
\end{eqnarray}

LIT can be written as:

\begin{eqnarray}
  LIT(P, Q)
  & = & \sum_{x \in X} lit(p_x, q_x)
\end{eqnarray}

which is a function of the natural logarithm. This is the result of the integral of any logarithm. Research has studied the Least Information Theory (LIT) and shown superior results in applications such as text clustering/classification and information retrieval \cite{Ke:2015,Ke:2017,Du:2018}.

\subsection{Entropy Discount}

We define the following entropy discount:

\begin{eqnarray}
  \Delta_H(P, Q)
  & = & \sum_{x \in X} \Big\lvert p_x - q_x \Big\rvert \frac{\int_{p_x}^{q_x} - p \log{p}\ dp}{\int_{p_x}^{q_x} x\ d x} \\
  % & = & \sum_{i=1}^{m} \abs{p_x - q_x} \frac{\frac{1}{4} \abs{p_x^2 (1 - 2\ln{p_x}) - q_x^2 (1-2\ln{q_x})}}{\frac{1}{2} \abs{p_x^2 - q_x^2}} \nonumber\\
  % & = & \sum_{i=1}^{m} \frac{\frac{1}{4} \abs{p_x^2 (1 - 2\ln{p_x}) - q_x^2 (1-2\ln{q_x})}}{\frac{p_x + q_x}{2}} \nonumber\\
  & = & \sum_{x \in X} \frac{\Big\lvert p_x^2 (1 - 2\ln{p_x}) - q_x^2 (1-2\ln{q_x}) \Big\rvert}{2(p_x + q_x)}
\end{eqnarray}

For any probabilities $p$ and $q$, let:

\begin{eqnarray}
\delta_h(p, q) & = & \frac{\Big\lvert p^2 (1 - 2\ln{p}) - q^2 (1-2\ln{q}) \Big\rvert}{2(p + q)}
\end{eqnarray}

The entropy discount $\Delta_H$ can be written as:

\begin{eqnarray}
  \Delta_H(P, Q) & = & \sum_{x \in X} \delta_h(p_x, q_x)
\end{eqnarray}

\subsection{DLITE: LIT with Entropy Discount}

We now define the Discounted Least Information Theory of Entropy (DLITE, pronounced as {\it delight}) as the amount of least information $LIT$ subtracted by its entropy discount $\Delta_H$:

\begin{eqnarray}
DL(P, Q)
& = & LIT(P, Q) - \Delta_H(P, Q) \\
& = & \sum_{x \in X} lit(p_x, q_x) - \delta_h(p_x, q_x)
\label{eq:dl}
\end{eqnarray}

% That is:
%
% \begin{eqnarray}
% &   & DL(P, Q) \nonumber \\
% & = & \sum_{i=1}^{m} \Big| p_x (1-\ln p_x) - q_x (1-\ln q_x) \Big| - \frac{\big| p_x^2 (1 - 2\ln{p_x}) - q_x^2 (1-2\ln{q_x}) \big| }{2(p_x + q_x)}
% \label{eq:dl_details}
% \end{eqnarray}

For any probability change from $p$ to $q$, let:

\begin{eqnarray}
dl(p, q) & = & lit(p, q) - \delta_h(p, q)
\label{eq:dl_x}
\end{eqnarray}

Equation~\ref{eq:dl} can written as:

\begin{eqnarray}
DL(P, Q) & = & \sum_{x \in X} dl(p_x, q_x)
\label{eq:dl_sum}
\end{eqnarray}

\section{DLITE and Properties}

Again, DLITE is the amount of Least Information (LIT) with the $\Delta_H$ discount:

\begin{eqnarray}
  DL(P, Q)
  & = & LIT(P, Q) - \Delta_H(P, Q) \\
  & = &  \sum_{x \in X} \int_{p_x}^{q_x} \log{\frac{1}{p}}\ dp - \sum_{x \in X} \Big\lvert p_x - q_x \Big\rvert  \frac{\int_{p_x}^{q_x} p \log{\frac{1}{p}}\ dp}{\int_{p_x}^{q_x} p\ d p}
\label{eq:dl2}
\end{eqnarray}

Whereas LIT represents the sum of weighted, microscopic entropy changes, it consists of an amount of entropy change due to the scale of related probabilities. This has led to the undesirable consequence of having different LIT amounts in different sub-system breakdowns.

The entropy discount $\Delta_H$ accounts for this unnecessary, extra amount in the LIT and reduces it to a scale-free measure. As shown in Equation~\ref{eq:dl2}, the discount on each $x$ dimension is a product of the absolute probability change in $p$ and the mean of $\log{\frac{1}{p}}$, which is subject to the scale of $p$ values.

\subsection{Metric Properties}

\begin{figure}[htb]
  \centering
  \includegraphics[width=2.5in]{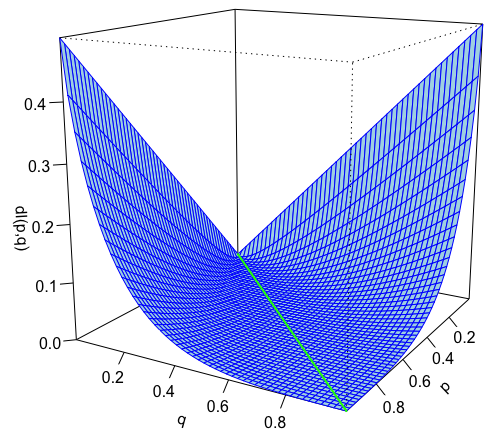}
  \caption{$dl(p,q)$ for any $p$ and $q$ values}
  \label{fig:pq}
\end{figure}

Given the definition in Equation~\ref{eq:dl} or \ref{eq:dl2}, it can be shown that DLITE satisfies the following metric properties:

\begin{enumerate}
  \item Non-negativity: $DL(P,Q) \ge 0$ for any probability distributions $P$ and $Q$ of the same dimensionality. See Appendix for proof.
  \item Identity of Indiscernibles: $DL(P,Q) = 0$ if and only if $P$ and $Q$ are identical distributions.
  \item Symmetry: $DL(P,Q) == DL(Q, P)$, the amount of the information from $P$ to $Q$ is the same as that from $Q$ to $P$.
\end{enumerate}

Figure~\ref{fig:pq} plots the value of dlite, $dl(p,q)$, for any probability change from $p$ to $q$ and demonstrates the above three properties: (1) all values $\ge 0$, (2) $0$ values only on the diagonal line where $p=q$, and (3) the symmetry indicating $dl(p,q)=dl(q,p)$.

While DLITE does not satisfy triangular inequality, its cube root $DLITE^{\frac{1}{3}}$ does:

\begin{eqnarray}
  \sqrt[3]{DL(P,Q)} + \sqrt[3]{DL(Q,R)} \ge \sqrt[3]{DL(P,R)}
\end{eqnarray}

where $P$, $Q$, and $R$ are probability distributions of the same dimensionality.

Given the above properties of DLITE, it is straightforward to show that $DLITE^{\frac{1}{3}}$ also satisfies non-negativity, identity of indiscernibles, and symmetry, and is, therefore, a metric. Because its cube root is a metric distance, DLITE can be regarded as a 3-dimensional {\it volumetric} measure in the amount of information. We refer to $DLITE^{\frac{1}{3}}$ as the {\it DLITE distance}.

This characteristic is similar to that of Jessen-Shannon (JS) Divergence, of which the square root is a metric \cite{Lin:1991,Endres:2006}. DLITE shares similar patterns with JS divergence in the measured amount of information.

\begin{figure}[htb]
  \begin{tabular}{cc}
  \includegraphics[width=2.3in]{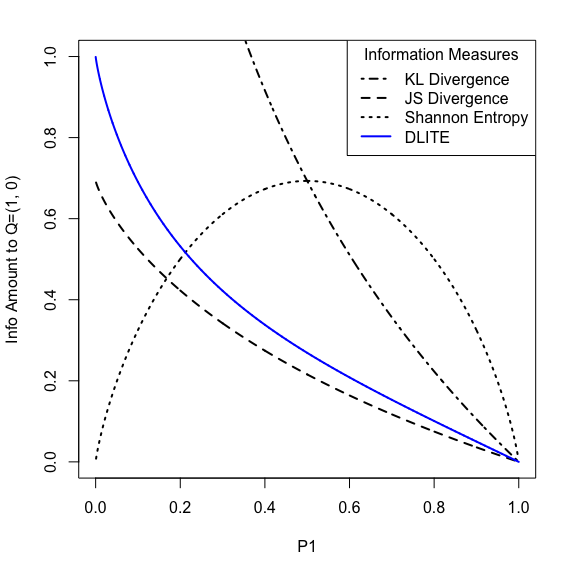} &
  \includegraphics[width=2.3in]{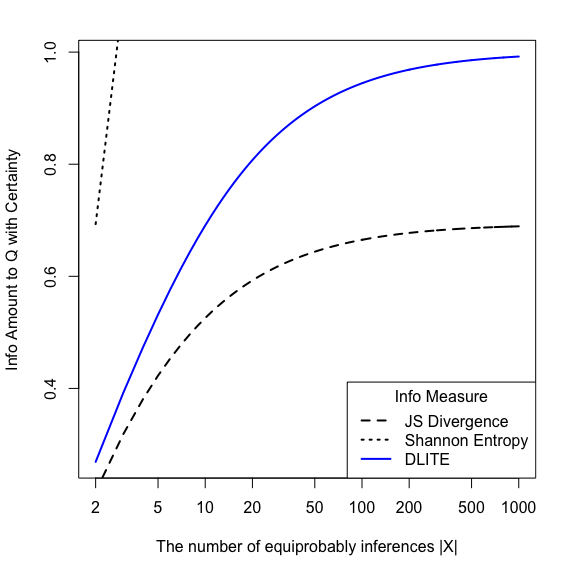} \\
  (a) $P(p_1,p_2) \to Q(1,0)$ &
  (b) Equiprobable $P$ to certainty $Q$
  \end{tabular}
\caption[DLITE vs. classic information measures on reducing to certainty]{DLITE vs. classic information measures on reducing to certainty. $Y$ is the amount of information $I(P,Q)$ based on each information measure. (a) is the binary case, where $X$ denotes probability $p_1$ of two mutually exclusive inferences, with $p_2 = 1 - p_1$. (b) shows the general case of reducing an equiprobable distribution $P$ to certainty $Q$, where log-transformed $X$ denotes the number of equiprobable inferences. }
\label{fig:certainty}
\end{figure}

Figure~\ref{fig:certainty} compares DLITE with classic measures including Shannon Entropy\cite{Shannon:1948}, KL divergence\cite{Kullback:1959}, and JS divergence\cite{Lin:1991} on reducing a probability distribution to certainty (when one inference becomes the ultimate outcome). Figure~\ref{fig:certainty} (a) compares the measures of reducing a binary probability distribution $P(p_1, p_2)$ to certainty $Q(1,0)$, i.e. with the first inferences as the ultimate outcome. Shannon entropy is symmetric here because it only accounts for the overall entropy reduction and disregards the amount of probability change in specific inferences. DLITE and Jessen-Shannon divergence follows a similar pattern with a bound whereas the KL divergence is unbounded, $KL \to \infty$ with $p_1 \to 0$ becomes the outcome.

In Figure~\ref{fig:certainty} (b), we compare the information measures when reducing equiprobable inferences to certainty. With an increasing number of equiprobable inferences, Shannon entropy continues to increase whereas DLITE and JS divergence are bounded. DLITE approaches $1$ when a large number of equiprobable inferences are reduced to certainty.

\subsection{Properties as an Information Quantity}

\begin{figure}[htb]
  \begin{tabular}{cc}
  \includegraphics[width=2.3in]{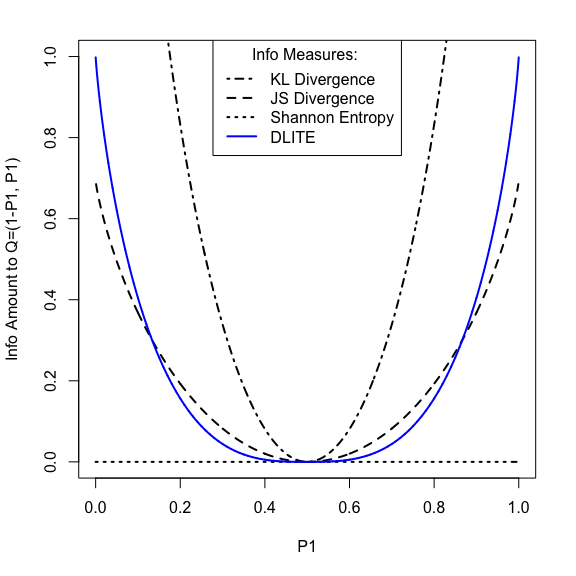} &
  \includegraphics[width=2.3in]{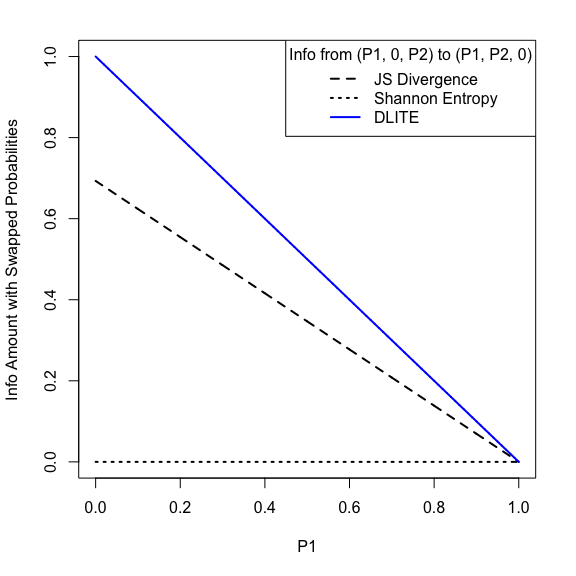} \\
  (a) Binary swap $P(p_1,p_2) \to Q(p_2,p_1)$ &
  (a) Swap $2$ of $3$ probabilities
  \end{tabular}
\caption[Information for Swapped Probabilities]{Information for Swapped Probabilities. $X$ axis denotes the probability of one inference whereas $Y$ shows the amount of information, with (a) swapped probabilities in the binary case and (b) swapped probabilities of $3$ inferences. Compare to Fig. 1 in \cite{Lin:1991}. }
\label{fig:swap}
\end{figure}

\subsubsection{Greater DLITE for More Equiprobable to Certainty}

% With more equiprobable inferences, a greater amount of DLITE is required to reduce the probability distribution to certainty. When the number of inferences approaches infinity, the required amount of DLITE is $1$.

With $|X|$ equiprobable inferences $p_x = \frac{1}{|X|}$, the amount of DLITE required to reduce the distribution to certainty (e.g. with one inference being the ultimate outcome $q_1=1$) increases with a growing number of inferences $|X|$ and asymptotically approaches $1$ with an infinite number of inferences:

\begin{eqnarray}
DL_{|X| \to \infty}\Big(P\big(p_x=\frac{1}{|X|}\big),Q_{certainty}\Big) = 1
\end{eqnarray}

% This is shown in Figure~\ref{fig:certainty} (b), where Jensen-Shannon Divergence, based on the natural logarithm, shares a similar pattern.
Whereas KL divergence is unbounded, Shannon entropy always increases with a greater number of equiprobable inferences. As shown in Figure~\ref{fig:certainty} (b), both Jensen-Shannon divergence and DLITE are bounded on reducing equiprobable probabilities to certainty. % DLITE approaches its maximum at $1$ with an infinite number of equiprobable inferences.

\subsubsection{DLITE Maximum $\le 1$}

% TODO need proof

DLITE is bounded in $[0,1]$ regardless of the dimensionality. $DLITE$ on one single inference $x \in X$ is maximized, $dl(p_x,q_x)=0.5$ when the probability changes from $p_x=0$ to $q_x=1$, or from $p_x=1$ to $q_x=0$. With $2$ mutually exclusive inferences, the overall DLITE is maximized for changes from $P=(0,1)$ to $Q=(1,0)$, where $DL(P,Q)=1$.

Shannon entropy, on the other hand, always returns $0$ when probabilities are swapped, as shown by examples in Figure~\ref{fig:swap}. In Figure~\ref{fig:swap} (a), KL Divergence approaches infinity with a $0$ probability whereas DLITE and JS divergence are bounded by $1$ and $\ln{2}$ respectively. Likewise, as Figure~\ref{fig:swap} (b) shows, DLITE is bounded in $[0,1]$ with 2 out of 3 probabilities swapped.

\begin{figure}[hbt]
  \begin{tabular}{cc}
  \includegraphics[width=2.3in]{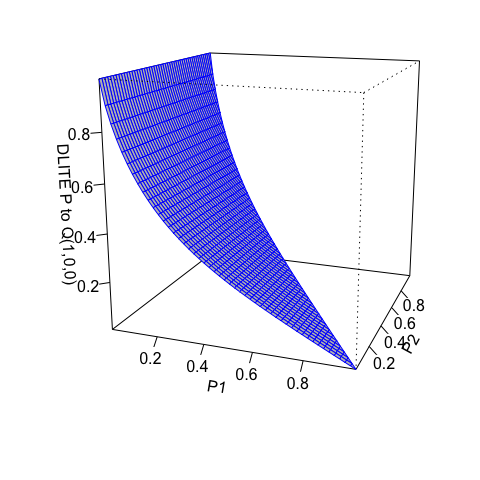} &
  \includegraphics[width=2.3in]{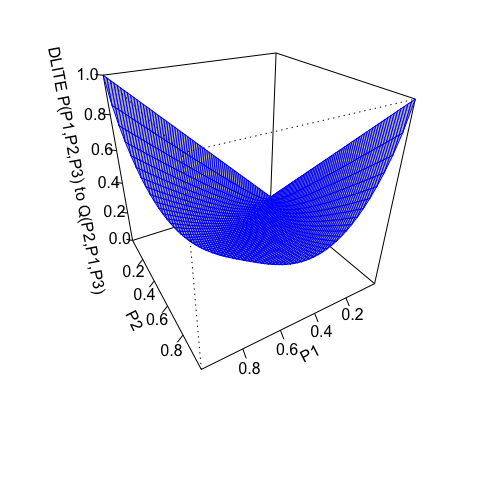} \\
  (a) $P(p_1,p_2,p_3) \to Q(1,0,0)$ &
  (b) $P(p_1,p_2,p_3) \to Q(p_2,p_1,p_3)$
  \end{tabular}
\caption[DLITE on $3$ mutually exclusive inferences]{DLITE on $3$ mutually exclusive inferences. $X$ and $Y$ denotes probabilities $p_1$ and $p_2$ of three mutually exclusive inferences, with $p_3 = 1 - p_1 - p_2$. $Z$ (vertical) is the DLITE quantity $DL(P,Q)$.}
\label{fig:3d}
\end{figure}

Figure~\ref{fig:3d} shows the function surface of DLITE on three inferences. Again, in all these cases, DLITE remains in the $[0,1]$ range. In Figure~\ref{fig:3d} (a), the probability distribution changes from $P$ (with the X coordinate for $p_1$ and Y for $p_2$) to $Q$, where the first inference is the outcome $p_1=1$. Figure~\ref{fig:3d} (b) shows the situations in which the probabilities are swapped, i.e. from $P(p_1,p_2,p_3)$ to $Q(p_2,p_1,p_3)$. It exhibits a symmetry and indicates that swapping the probabilities in the opposite direction results in the same amount of DLITE.

\subsubsection{Overall DLITE as Weighted Sum of Sub-systems}

% For a binary probability distribution $(a,b)$, where $a+b=1$ and $b$ can be further broken down as a subsystem with internal probabilities of $X$. Suppose the subsystem distribution changes from $P$ to $Q$, where $\sum_{x \in X} p_x =1$ and $\sum_{x \in X} q_x =1$, it can be shown that:
%
% \begin{eqnarray}
%   DL(P(a,bX),Q(a,bX))
%   & = & DL(ab, ab) + b DL(P_X, Q_X) \\
%   & = & b DL(P_X, Q_X)
% \end{eqnarray}

Suppose each inference $x \in X$ can be broken down into a subsystem of mutually exclusive inferences $s \in x$, where the sum of their probabilities:

\begin{eqnarray}
  \sum_{x \in X} p_{x} & = & 1 \\
  \sum_{s \in x} p_{s} & = & 1
\end{eqnarray}

In the overall system $X_S$ of all sub-systems combined, the probability of each inference of $x_s \in x$ of $P$ distribution is:

\begin{eqnarray}
p_{x_s} & = & p_x p_s
\end{eqnarray}

And the sum of their probabilities:

\begin{eqnarray}
  \sum_{x \in X} \sum_{s \in x} p_x p_s & = & 1
\end{eqnarray}

Assume the $P$ distribution for $X$ remains unchanged, hence:

\begin{eqnarray}
DL(P_X, P_X) & = & 0
\end{eqnarray}

Suppose the sub-system distributions change from $P$ to $Q$, then the DLITE of the $x$ sub-system is:

\begin{eqnarray}
  DL(P_{x}, Q_{x})
  & = & \sum_{x_s \in x} dl(p_{x_s}, q_{x_s})
\end{eqnarray}

The overall DLITE for $X_S$ can be computed by:

\begin{eqnarray}
  DL(P_{X_S}, Q_{X_S})
  & = & \sum_{x \in X} \sum_{x_s \in x} dl(p_x p_{x_s}, p_x q_{x_s}) \\
  & = & \sum_{x \in X} p_x \sum_{x_s \in x} dl(p_{x_s}, q_{x_s}) \\
  & = & \sum_{x \in X} p_x DL(P_x, Q_x) \\
  & = & \sum_{x \in X} p_x DL(P_x, Q_x) + DL(P_X, P_X)
\end{eqnarray}

In other words, DLITE of the overall system $X_S$ can be computed by the weighted sum of DLITE amounts for $x \in X$ sub-systems. See Appendix for proof of $dl(xp, xq) = x \cdot dl(p,q)$, which leads to the sub-system breakdown rule here as well as the following properties of product and joint probability distributions.

\subsubsection{Independent X and Y}

For variables $X$ and $Y$ that are statistically independent, the joint probability of $x$ and $y$ can be computed by:

\begin{eqnarray}
p_{xy} & = & p_x p_y
\end{eqnarray}

Let $P_X P_Y$ be the joint probability distribution of the two distributions $P_X$ and $P_Y$. Assume the probability distribution of $X$ changes from $P_X$ to $Q_X$ and the distribution of $Y$ remains $P_Y$, it can be shown that:

\begin{eqnarray}
DL(P_X P_Y, Q_X P_Y)
& = & DL(P_X, Q_X) \\
& = & DL(P_X, Q_X) + \underbrace{DL(P_Y,P_Y)}_{\text{=0, no change}}
\end{eqnarray}

\subsubsection{Dependent X and Y}

For dependent variables $X$ and $Y$, the joint probability of $x$ and $y$ can be computed by:

\begin{eqnarray}
p_{xy} & = & p_x p_{y|x}
\end{eqnarray}

Let $P_{XY}$ be the joint probability distribution of the two distributions $P_X$ and $P_Y$. $Q_{XY}$ is the changed joint distribution.

If the probability distribution of $X$ changes from $P_X$ to $Q_X$ whereas the conditional distribution is unchanged with $P_{Y|X}$, it can be shown that:

\begin{eqnarray}
DL(P_{XY}, Q_{XY})
& = & DL(P_X,Q_X) \\
& = & DL(P_X,Q_X) + \underbrace{DL(P_{Y|X},P_{Y|X})}_{\text{=0, no change}}
\end{eqnarray}

If the probability distribution of $X$ remains unchanged at $P_X$ and the conditional distribution changes from $P_{Y|X}$ to $Q_{Y|X}$, it can be shown that:

\begin{eqnarray}
DL(P_{XY}, Q_{XY})
& = & DL(P_{Y|X}, Q_{Y|X}) \\
& = & DL(P_{Y|X},Q_{Y|X}) + \underbrace{DL(P_X,P_X)}_{\text{=0, no change}}
\end{eqnarray}

% \section{Studies and Results}
%
% \subsection{Clustering and Classification}
%
% \subsection{Information Retrieval}

\section{Conclusion}

The proposed DLITE measure exhibits a set of very useful characteristics. It meets the metric properties of non-negativity, identity of indiscernibles, and symmetry. Additionally, its cube root $DLITE^{\frac{1}{3}}$ satisfies the property of triangular inequality and is a metric distance.

DLITE also manifests several other desirable properties of an information measure. Its value is bounded in $[0,1]$, increases with more equiprobable inferences reduced to a certainty, and can be computed as the weighted sum of DLITE in the sub-systems. DLITE is additive in cases of dependent and independent variables. These properties support the use of DLITE in applications where the amount of information is to be measured and aggregated properly.

\appendix
\section*{Appendix}
\input{proof_nonneg}

\input{proof_uniform}

\input{proof_scale}

% \bibliographystyle{plain}
% \bibliography{lit}

\end{document}

%% file: proof_nonneg.tex
% !TEX root = dlite_2019.tex

% \section{DLITE Nonnegativity}

% \begin{definition}{DLITE}
%
% \end{definition}

\begin{theorem}

% Given the following DLITE definition,
%   \begin{eqnarray}
%   DL(P, Q)
%   & = &
%         \sum_{i=1}^{m} \abs{ p_i (1-\ln p_i) - q_i (1-\ln q_i) } \\
%   & - & \sum_{i=1}^{m} \frac{\abs{p_i^2 (1 - 2\ln{p_i}) - q_i^2 (1-2\ln{q_i})}}{2(p_i + q_i)}
%   \end{eqnarray}
%
% where $p_i$ and $q_i$ are probabilities of the $i^{th}$ inferences on distributions $P$ and $Q$ respectively:

For any probability distributions $P$ and $Q$:
\begin{eqnarray}
DL(P, Q) \ge 0
\end{eqnarray}

\end{theorem}

\begin{proof}

The DLITE equation can be written as:

\begin{eqnarray*}
&   & DL(P \to Q) \nonumber\\
& = & LIT(P \to Q) - \delta_H(P \to Q) \nonumber\\
& = & \sum_{p_i \ge q_i}^{} p_i (1-\ln p_i) - q_i (1-\ln q_i)  - \frac{p_i^2 (1 - 2\ln{p_i}) - q_i^2 (1-2\ln{q_i})}{2(p_i + q_i)} \nonumber\\
& + & \sum_{p_i<q_i}^{} q_i (1-\ln q_i) - p_i (1-\ln p_i)  - \frac{ q_i^2 (1-2\ln{q_i}) - p_i^2 (1 - 2\ln{p_i})}{2(p_i + q_i)}
\end{eqnarray*}

For any values $c$ and $x \ge c$, let:

\begin{eqnarray*}
g(x, c) & := & x (1-\ln x) - c (1-\ln c)  - \frac{x^2 (1 - 2\ln{x}) - c^2 (1-2\ln{c})}{2(x + c)}
\end{eqnarray*}

DL can be rewritten as:

\begin{eqnarray}
DL(P, Q) & = & \sum_{p \ge q}^{} g(p, q) + \sum_{p < q}^{} g(q, p)
\label{eq:dl_g}
\end{eqnarray}

The derivative of $g(x,c)$ with regard to $x$ is:

\begin{eqnarray*}
g'(x) & = & -\dfrac{2c^2\ln\left(x\right)-x^2-2c^2\ln\left(a\right)+c^2}{2\left(x+c\right)^2}
\end{eqnarray*}

The minimum of $g(x)$ can be obtained at $g'(x)=0$:

\begin{eqnarray*}
&   & 2c^2\ln\left(x\right)-x^2-2c^2\ln\left(a\right)+c^2 \\
& = & x^2 - c^2(1+2\ln\frac{x}{c}) \\
& = & 0
\end{eqnarray*}

That is:

\begin{eqnarray*}
x^2 & = & c^2(1+2\ln\frac{x}{c}) \\
\ln\frac{x}{c} & = &  \frac{(\frac{x}{c})^2 - 1}{2}
\end{eqnarray*}

Let $r=\frac{x}{c}$, this becomes:

\begin{eqnarray*}
\ln r & = &  \frac{r^2 - 1}{2}
\end{eqnarray*}

Or:

\begin{eqnarray*}
r & = &  e^\frac{r^2 - 1}{2}
\end{eqnarray*}

The only solution is $r=\frac{x}{c}=1$, i.e. $x=c$. Hence the minimum of $g(x)$ is at $x=q$, where $g(x)=0$. Therefore, $g(x) \ge 0$, with the zero value at $x=c$. Based on Equation~\ref{eq:dl_g}, where DLITE is the sum of $g(p,q)$ and $g(q,p)$, we conclude that $DL(P, Q) \ge 0$.

\end{proof}

%% file: proof_uniform.tex
% !TEX root = dlite_2019.tex

\begin{theorem}

With $|X|$ equiprobable inferences $p_x = \frac{1}{|X|}$, the amount of DLITE required to reduce the distribution to certainty (e.g. with one of the $p_x=1$) increases with the increase in the number of inferences.

\end{theorem}

\begin{proof}

  Suppose $X$ has $m$ mutually exclusive inferences that are equally likely, i.e. $p_x=\frac{1}{m}$, $\forall x \in X$. The amount of DLITE to reach certainty -- that is, one inference becomes the ultimate outcome $q_1=1$ -- is:

  \begin{equation}
    DL(m) =
  \left(\dfrac{1}{m\left(m+1\right)}-\dfrac{1}{m}\right)\ln\left(m\right)+\dfrac{\frac{1}{m}-m}{2\left(m+1\right)}+\dfrac{3(m-1)}{2m}-
  \end{equation}

  Its derivative is:

  \begin{equation}
  DL'(m) = \dfrac{m^2\ln\left(m\right)+m+1}{m^2\left(m+1\right)^2}
  \end{equation}

  $DL'(m)$ is always positive, decreases when $m$ increases, and approaches zero with an infinite number of inferences, i.e. $DL'_{m \to \infty}(m) = 0$. In other words, $DL(m)$ always increases with the greater number of mutually exclusive and equally likely inferences. It approaches its maximum when $m \to \infty$ where $DL_{m \to \infty}(m) = 1$.

\end{proof}

%% file: proof_scale.tex
% !TEX root = dlite_2019.tex

\begin{theorem}

Given the $dl(p,q)$ function of probability change from $p$ to $q$ in Equation~\ref{eq:dl_x}, for any positive value $x$:

\begin{eqnarray}
  dl(xp, xq) & = & x \cdot dl(p, q)
\end{eqnarray}

\end{theorem}

\begin{proof}

If $p \ge q$, then $xp \ge xq$.

The amount of dlite for the scaled values from $xp$ to $xq$ is:

\begin{eqnarray*}
  dl(xp, xq)
  & = & \Big(xp (1-\ln xp) - xq (1-\ln xq) \Big) \\
  &   & - \frac{\Big( x^2p^2 (1 - 2\ln{xp}) - x^2q^2 (1-2\ln{xq}) \Big)}{2x(p + q)} \\
  & = & x \Big(p (1-\ln p) - q (1-\ln q) \Big) \\
  &   & - x\frac{\Big( 2p^2 (1 - 2\ln{p}) - q^2 (1-2\ln{q}) \Big)}{2(p + q)} \\
  &   & - (p - q)x\ln{x}) + \frac{2x^2p^2\ln{x} - 2x^2q^2\ln{x}}{2x(p+q)} \\
  & = & x \cdot dl(p,q)
\end{eqnarray*}

If $p<q$, the same can be obtained.

\end{proof}